\documentclass[11pt, a4paper]{article}

\usepackage{titlesec,soul}
\usepackage{subcaption}

 \usepackage[margin=1.2in]{geometry}

\usepackage[
			anchorcolor=black,
			citecolor=black,
			colorlinks=true,
			filecolor=black,
			linkcolor=black,
			menucolor=black,
			runcolor=black,
			urlcolor=black,
			hyperfootnotes=false
]{hyperref}
\usepackage{amsthm,amsmath,microtype}
\usepackage[noabbrev]{cleveref}

\theoremstyle{plain}

\newtheorem{proposition}{Proposition}
\newtheorem{corollary}{Corollary}
\theoremstyle{definition}
\newtheorem{definition}{Definition}

\newtheorem{assumption}{Assumption}

\newtheorem{example}{Example}
\theoremstyle{remark}

\crefname{observation}{Observation}{Observations}
\crefname{theorem}{Theorem}{Theorems}
\crefname{lemma}{Lemma}{Lemmata}
\crefname{proposition}{Proposition}{Propositions}
\crefname{corollary}{Corollary}{Corollaries}
\crefname{definition}{Definition}{Definitions}
\crefname{assumption}{Assumption}{Assumptions}
\crefname{section}{Section}{Sections}
\crefname{figure}{Figure}{Figures}
\crefname{axiom}{Axiom}{Axioms}
\crefname{modification}{Modification}{Modifications}
\crefname{example}{Example}{Examples}
\crefname{appendix}{Appendix}{Appendices}

\usepackage[font={small}]{caption}[2008/08/24] 

\usepackage{enumitem} 

\usepackage{setspace}

\usepackage{amsmath,amsfonts,amssymb, mathrsfs}
\usepackage[T1]{fontenc}
\usepackage[]{lmodern}
\usepackage[utf8]{inputenc}
\usepackage[ngerman, english]{babel}
\usepackage{csquotes} 
\usepackage{etex} 
\usepackage{multirow,array}

\usepackage[noabbrev]{cleveref}

\usepackage{pgfplots}
\pgfplotsset{compat=newest}
\usepackage{sidecap} 
\usetikzlibrary{patterns,decorations.pathreplacing,positioning,external,shapes.geometric,arrows,calc}

\usepackage{pstricks} 
\usepackage{pst-3d}
\usepackage{sgamevar, egameps} 
\usepackage{graphicx}
\usepackage{epstopdf} 

\usepackage[backend=biber,
			sorting=ynt,
			sortcites=true,
			bibencoding=inputenc,
			bibstyle=authoryear -comp,
			citestyle=authoryear -comp,
			doi=false,
			giveninits=true,
			isbn=false,
			uniquename=false,
			uniquelist=minyear,
			maxcitenames=3,
			maxnames=3,
			maxbibnames=99,
			natbib=true,
			uniquename=false,
			url=false,
			date=year
]{biblatex}
\bibliography{mybib}

\DeclareSourcemap{
  \maps[datatype=bibtex]{
    \map[overwrite=true]{
      \step[fieldsource=number]
      \step[fieldset=issue, origfieldval]
    }
  }
}

	\AtEveryBibitem{\clearlist{language}}
	\AtEveryBibitem{\clearfield{number}}
	\AtEveryBibitem{\clearlist{month}}
	\AtEveryBibitem{\clearfield{Month}}
	\AtEveryBibitem{\clearfield{day}}
	\AtEveryBibitem{\clearfield{eprint}}

 \DeclareSourcemap{
  \maps[datatype=bibtex,overwrite=true]{
    \map{
      \step[fieldsource=pages,
            match=\regexp{pp?\.?(.+)},
           replace=\regexp{$1}] 
    }
  }
}

\renewbibmacro{in:}{}
\usepackage{verbatim}
\usepackage{booktabs} 
\usepackage[multiple]{footmisc}

\usepackage{placeins} 

\definecolor{Vermilion}{RGB}{213,94,0}
\definecolor{BluishGreen}{RGB}{0,158,115}
\definecolor{Cadmium}{RGB}{227, 0, 34}
\definecolor{ReddishPurple}{RGB}{204,121,167}
\definecolor{Orange}{RGB}{230,159,0}
\definecolor{Yellow}{RGB}{240,228,66}
\definecolor{SkyBlue}{RGB}{86,180,233}
\definecolor{Blue}{RGB}{0,114,178}

\DeclareFontFamily{U}{mathx}{\hyphenchar\font45}
\DeclareFontShape{U}{mathx}{m}{n}{
      <5> <6> <7> <8> <9> <10>
      <10.95> <12> <14.4> <17.28> <20.74> <24.88>
      mathx10
      }{}
\DeclareSymbolFont{mathx}{U}{mathx}{m}{n}
\DeclareFontSubstitution{U}{mathx}{m}{n}
\DeclareMathAccent{\widecheck}{0}{mathx}{"71}
\DeclareMathAccent{\wideparen}{0}{mathx}{"75}

\setlist{noitemsep,topsep=1pt}
\title{Mechanism Design with Informational Punishment\thanks{We thank the advisory editor and two anonymous referees for their help in improving the paper. Johannes Schneider gratefully acknowledges financial support from the German Research Foundation (DFG) through CRC TR 224 (Project B03), Agencia Estatal de Investigación (grant PID2019-111095RB-I00 and grant PID2020-118022GB-I00), Ministerio Economía y Competitividad (grant ECO2017-87769-P), and Comunidad de Madrid (grant MAD-ECON-POL-CM H2019/HUM-5891).}}
\author{Benjamin Balzer\thanks{University of Technology Sydney, benjamin.balzer@uts.edu.au} \and Johannes Schneider\thanks{Universidad Carlos III de Madrid, jschneid@eco.uc3m.es}}
\allowdisplaybreaks
\begin{document}
\maketitle
	\begin{abstract}
		 We introduce \emph{informational punishment} to the design of mechanisms that compete with an exogenous status quo mechanism: Players can send garbled public messages  with some delay and others cannot commit to ignoring them. Optimal informational punishment ensures that  full participation is without loss, even if any single player can publicly enforce the status quo mechanism. Informational punishment permits using a standard revelation principle, is independent of the mechanism designer's objective, and operates exclusively off the equilibrium path. It is robust to refinements and applies in informed-principal settings. We provide conditions that make it robust to opportunistic signal designers.
	\end{abstract}
\newpage
\section{Introduction} 
\label{sec:introduction}
 
Economics aims at improving the efficiency of institutions to govern strategic parties. If those parties have asymmetric information, the first welfare theorem fails and we may fail to attain efficiency through any institution. In these cases, mechanism design becomes a powerful tool. Mechanism design characterizes all outcomes that parties' strategic incentives permit. Per the revelation principle, a simple direct revelation mechanism is enough to determine the best results we can hope for under any institution. In reality, however, there may be an additional friction: new institutions often replace status quo mechanisms.

Changing to the new institution may require consent: if a party vetoes the proposal, the status quo prevails. Parties follow similar strategic incentives in their decision to veto the new institution to those at play when they decide on their behavior within any given institution. Vetoing a proposal can be strategic, too. A party may publicly veto the mechanism only to signal her private information.

This paper considers a setting in which parties may veto a proposed mechanism and disclose their veto to others. If the proposal fails, parties revert to the play of a default game. The optimal mechanism in such settings may involve on-path rejections \citep{Celik:11}. Thus, it might not be optimal to offer a mechanism attractive to everyone. As a result, complexity increases, and the mechanism-design approach loses part of its power. We show, however, that if parties can store information and commit to releasing it at a later date, assuming full participation is without loss, and we can readily apply the tools of mechanism design. That is, the optimal mechanism is one that is accepted by all parties and types. We refer to this technology as \emph{informational punishment}. The information is stored and released in the event of a deviation. Its purpose is to punish the deviator by the release.

Informational punishment is a simple yet powerful tool. It requires only that every party has access to a simpel signaling device with the following feature: the device can conceal the information for some time before releasing a garbled version of it. The threat to release information later suffices to discipline others and ensures full participation. Because informational punishment operates only off the equilibrium path, it does not interfere with the design of the mechanism itself. Moreover, a decentralized implementation through parties themselves is straightforward. Informational punishment is robust both to equilibrium refinements and restrictions on the space of available mechanisms. In addition, it applies to informed-principal problems and is often immune to a designer who suffers from informational opportunism.

Examples of our environment abound. Parties in a legal conflict can coordinate to settle through an arbitration mechanism. However, each party can unilaterally enforce the default game of a trial. Political parties can ease gridlock through a bargaining procedure. However, they might refuse to cooperate and instead enter a stalemate until the next election. Firms can work together to determine the standard of the industry. However, they can also launch a standards war. Countries can negotiate free trade agreements. However, if one government refuses to engage in those negotiations, all countries will revert to the World Trade Organization trade regime.

In all of these cases, vetoing the mechanisms may signal a party's ability in the default game. The other parties interpret that signal and adjust their behavior. The change in behavior influences the default game's outcomes. Moreover, if vetoing signals information, participating signals information too. Thus, if we cannot rule out on-path rejections of the mechanism, we need to solve all combinations of vetoing and participating and the associated outcomes of the default game to compute the optimal mechanism---a cumbersome computational problem. 

Informational punishment is readily available in these cases. Disputing parties can signal their private information through their attorneys, who they instruct to release parts of the information only if the settlement fails; political parties can leak information to journalists who require time to fact-check it; firms can publish beta versions or announce products that turn out to be vaporware; and countries can outsource an information campaign to a  state bureaucracy that needs time to organize it. All these options represent a form of decentralized informational punishment that meets the requirements for commitment to store partial information that will be released at a later date.

Adding informational punishment relaxes the computational burden to compute the optimal mechanism. It restores full participation on the equilibrium path yet only affects the players' outside options. Incentives inside the mechanism remain unchanged. Thus, we relax participation constraints without affecting incentive constraints. Moreover, the off-path event of informational punishment solves a specific, well-defined information-design problem: min-max the deviator's continuation payoff over Bayes-plausible information structures.

Informational punishment works because a signal realization about party $i$ has two effects. The first effect is direct and distributional: the other parties update expected payoffs because $i$'s private information is payoff relevant. The second effect is indirect and behavioral. A party's strategy is a function of the information set. Altering information alters the continuation strategy. Via equilibrium reasoning, the party's change in behavior alters the behavior of other parties too. Informational punishment exploits that channel.

\paragraph{Related Literature.} The literature on signaling through vetoes in mechanisms dates back to \citet{Cramton:95}. Like them, we consider a problem in which a proposed mechanism competes with a status quo game. In line with their model, we assume that ratification is public, so parties learn who vetoes the mechanism and who does not. Unlike them, we are not interested in only describing the mechanism space that is ratifiable by everyone.

Instead, and closer to \citet{Celik:11}, we assume perfect Bayesian equilibrium as our solution concept, take the space of mechanisms as given, and are interested in the set of outcomes that can arise in such a setting. Unlike \citet{Celik:11}'s setting, in which equilibrium rejection may be optimal, we allow parties to engage in informational punishment. We show that equilibrium rejection is of no concern when informational punishment is available. The set of outcomes that we can implement with full participation contains those we can implement with equilibrium rejection.\footnote{\citet{TAN2007383,DEQUIEDT2007302} are other examples of default games' threat to participation.}

\Citet{gerardi2007sequential,correia2017trembling} propose an alternative approach in settings with veto-constrained mechanisms. Instead of signaling information, they consider mechanisms that ``tremble.'' Even if all parties decide to participate, the mechanism breaks down with a small probability and invokes the default game. The mechanism can get arbitrarily close to the full-participation optimum through such trembling. However, these mechanisms rely on the assumption that a deviating party cannot credibly signal that it was her veto that invoked the default game and not the mechanism's tremble. Instead, we allow parties to announce their veto publicly, making trembles insufficient to overcome the veto problem. \Citet{Celik:13} use reciprocal mechanisms to circumvent the equilibrium-rejection problem.
The main difference from our approach is that parties make a public revelation on the equilibrium path through their mechanism proposal. Depending on the environment, these revelations can interfere with incentive compatibility. Because informational punishments affect only off-path events, these concerns are absent; the set of implementable outcomes nests those in \citet{Celik:13}. The reverse, however, is not the case.

Our previous work on standard-setting organizations \citep{Balzer:15} applies informational punishment outside mechanism design. In that paper, we consider a setting where the available mechanisms contain only efficient take-it-or-leave-it offers with fixed shares. In \citet{Balzer:15} informational punishment enlarges the class of environments in which full participation is feasible, yet informational punishment cannot guarantee full participation. The reason is that the space of available mechanisms is too small. In the current paper, we take a broader view and complement \citet{Balzer:15} in two ways. First, we define the minimal set of available mechanisms in the designer's toolbox such that an optimal full-participation mechanism exists. Second, we show how informational punishment simplifies the mechanism designer's task, particularly when the set of available mechanisms is large. Unlike the results in \citet{Balzer:15}, our results readily apply to various design problems. Moreover, we show that---given our minimal conditions---common restrictions on the designer's problem do not affect the power of informational punishment. Specifically, we consider equilibrium-refinement concepts \citep{Cho:87,Cramton:95,grossman1986perfect}, informed-principal problems \citep{Myerson:83informed}, and informational opportunism \citep{MartimortDequiedt15}.

\section{Setup} 
	\label{sub:setup}
	\paragraph{Players and Information Structure.} There are $N$ players, indexed by $i  \in \mathcal N:= \lbrace 1,...,   N \rbrace$. Each player has a private type $\theta_i \in  \Theta_i$, and $\Theta_i \subset \mathbb{R}$ is compact. The state $ \theta :=\{\theta_1,..., \theta_N\} \in \Theta:=\times_i \Theta_i$ is distributed according to a commonly known distribution function $I^0: \Theta \rightarrow \Delta(\Theta)$, the \emph{prior information structure}. Let $\theta_{-i} := \theta \setminus \theta_i$, and define the marginal $I^0_i(\theta_i  ):= \int_{\Theta_{-i}} I^0(\theta_i,d\theta_{-i})$ with support $supp(I^0_i) = \Theta_i$.

	An information structure $I: \Theta \rightarrow \Delta(\Theta)$ is a commonly known joint distribution over the state $\theta$. The only restriction we impose on $I$ is that it is absolutely continuous w.r.t. $I^0$; that is, $ supp(I)\subseteq supp(I^0)$. Given $\theta_i$ a player's \emph{belief} about the other players' types is the conditional distribution $I_{-i}(\theta_{-i}|\theta_i):=\frac{I(\theta_i,\theta_{-i})}{I_i(\theta_i)}$, where $I_i(\theta_i)$ is the marginal of $I$. Let $\mathcal I^0$ be the set of all information structures for which $I^0$ is an expansion. That is, $I  \in \mathcal I^0$ if and only if there exists a random variable $\widetilde \Sigma: \Theta \rightarrow \Delta(S)$ that maps types to distributions of signals such that the realization $\sigma \in S$ together with $\widetilde \Sigma$ and $I^0 $ implies $I$ via Bayes' rule.

	\paragraph{Basic Outcomes, Decision Rules, and Payoffs.} There is an exogenously given set of basic outcomes, $Z \subset \mathbb{R}^K$, with $K < \infty$. Player $i$ values the outcome $z \in Z$ according to a Bernoulli utility function, $u_i$, defined over $Z \times \Theta$.

	We represent the rules of a game by a decision rule, \[\pi: \Theta \rightarrow \Delta(Z),\] where $\Delta(Z)$ is the set of all distribution functions over the outcome space $Z$. Each rule $\pi$ is a mapping from \emph{type reports} to a distribution over outcomes represented by the distribution function $G_\pi(z|\theta)$.

	\paragraph{Status quo.} The status quo is an exogenous game of incomplete information. We assume an equilibrium in that game exists for any information structure $I \in \mathcal I^0$ and take the equilibrium selection as given.\footnote{Selecting the designer's least preferred equilibrium of the default game provides the benchmark case in which equilibrium selection poses the smallest threat to participation. However, even in such cases, full participation may not be optimal absent informational punishment \citep[see, for example,][]{Celik:11}.} For any information structure $I$, the status quo induces a decision rule $\mathring{\pi}_I$. Under $\mathring{\pi}_I$, the expected utility of a truthfully reporting player $i$ with type $\theta_i$ is

	\begin{equation}\label{IIC}
	\begin{split}
		v_i(\theta_i,I,\mathring{\pi}_I)&:= \int_{\Theta_{-i}} \int_{Z} u_i({z},\theta_i,\theta_{-i}) dG_{\mathring{\pi}_I}(z|\theta_i,\theta_{-i}) I(d\theta_{-i}|\theta_i) \notag \\&=\max_{m_i \in \Theta_i}\int_{\Theta_{-i}} \int_{Z} u_i({z}, \theta_i,\theta_{-i})  dG_{\mathring{\pi}_I}(z|m_i,\theta_{-i})I(d\theta_{-i}|\theta_i),   \qquad (I\textit{-IC})
	\end{split}
	\end{equation}
	almost everywhere conditional on $I$; that is, $\forall \theta_i \in supp(I_i)$. The second line follows because $\mathring{\pi}_I$ is incentive compatible under information structure $I$ (henceforth $I$-IC).\footnote{We employ standard revelation arguments here. In equilibrium, the strategies of all players and types are common knowledge. Thus, each type can use a proxy to whom she truthfully feeds her type to play the game. We can thus describe the outcome of every game by an incentive-compatible decision rule conditional on the initial (common knowledge) information structure. See, for example, \citet[chapter 6]{myerson_game_1991} for further discussion.} Truthful reporting is optimal for all types of players given $\mathring{\pi}_I$.

	The existence of equilibrium under $\mathcal I^0$ implies that the collection of possible status quo outcomes $\mathring{\Pi}:=\{\mathring{\pi}_I\}_{I\in \mathcal I^0}$ (with $\mathring{\pi}_I$ being $I$-IC) is well defined. 

	\paragraph{Mechanism.} The mechanism is an alternative to the status quo. Any mechanism is a game of incomplete information represented by a decision rule. The collection of decision rules is $\Pi$. Given $\pi$ and $I$, we define each player's optimal reporting strategy $m_{i,I}(\theta_i)$. We collect players' reports in $m_I(\theta)$. An equilibrium of $\pi$ implements the decision rule $\pi_I:= \pi \circ m_I: \Theta \rightarrow \Delta(Z)$ which is $I$-IC.\footnote{Although any decision rule in $\Pi$ represents a direct revelation mechanism, a truthful implementation cannot be guaranteed. Indeed, $\Pi$ is shorthand for all available game forms without loss represented by those that have type reports as players' action choices. The equilibrium play of each $\pi \in \Pi$ under $I$ then induces some $I$-IC decision rule.}

	The set of available mechanisms, $\Pi$, may be restricted by legal or institutional constraints, or particular outcomes may simply be infeasible. We require two properties.
\begin{assumption}
The set of available mechanisms is such that
	\begin{enumerate}[label=(\roman*)]
		\item $\mathring{\Pi}  \subseteq \Pi$ and
		\item $\Pi$ is closed under convex combinations; that is, if $ \pi,\pi' \in \Pi$, then for any $\lambda: \Theta \rightarrow [0,1]$  it holds that $ \lambda \pi + (1- \lambda) \pi'=:\pi^\lambda \in \Pi$.
	\end{enumerate}
\end{assumption}

	The first property implies that the mechanism can replicate the status quo for any given $I\in \mathcal{I}^0$. The second property implies that if two games (A and B) are part of the available mechanisms, so is the game in which game A is played for specific type reports and game B for the remaining type reports.

	Apart from these requirements, we do not restrict $\Pi$. Instead, we allow for both a classical mechanism-design setting and the possibility that the designer's set of mechanisms is exogenously limited. The latter includes pure mediation within the status quo.

	\paragraph{Informational Punishment.} We assume all players have access to a signaling device $\Sigma$. The N-dimensional random variable $\Sigma: \Theta \rightarrow S$ maps type reports to realizations in signal space $S\equiv S_1 \times S_2 \times ..\times S_N$ with $|S_i| \geq |\Theta_i|$. We denote the realization of $\Sigma$ by $\sigma \in S$ and that of element $\Sigma_i$ by $\sigma_i \in S_i$. 

	\paragraph{Timing.} First, players learn their types and observe $(\pi,\Sigma)$. Second, they simultaneously send a message $m_i^\Sigma$ to $\Sigma$. Third, players simultaneously decide whether to veto the mechanism. If at least one player vetoes the mechanism, the set $V$ of vetoing players becomes common knowledge, and the signal realizations $\sigma$ become public. Players use that information to update to an information structure $I^{V,\sigma} \in \mathcal I^0$ and the status quo implements $\mathring{\pi}_{I^{V,\sigma}} $. If players unanimously ratify the mechanism, they report $m_i$ to the mechanism that implements $\pi$.

	\paragraph{Solution Concept and Veto Beliefs.} We consider all mechanism-signaling device combinations $(\pi,\Sigma)$ implementable as the grand game's perfect Bayesian equilibrium (PBE) using the definition from \citet{Fudenberg:88}.

	We use \emph{veto information structures}, $I^V$, to refer to the information structures that arise \emph{after} an observed veto but \emph{before} the realization $\sigma$. PBE implies that $I^V(\theta_{-i}|\theta_i) = I^{V\setminus i}(\theta_{-i}|\theta_i)$ for any $i \in V$ and $I^V(\theta_{-i}|\theta_i) = I^{V \cup i}(\theta_{-i}|\theta_i)$ for any $i \notin V$. In addition, all but first-node off-path beliefs about deviators follow Bayes' rule. The remaining off-path beliefs are arbitrary.\footnote{In our setting, a player is observed to deviate at most once. Off-path belief cascades (see \citet{sugaya2017revelation}) are thus not possible in our model.}

\subsection{Discussion} 
\label{sub:discussion_of_the_modeling_assumptions}

	Before moving to the analysis, we pause to discuss the key elements of our rather abstract model.

    \paragraph{Timing and Information Flow.} It is useful to assume---for the moment---that the mechanism designer only selects $\pi$. In addition, there is a signal designer who selects $\Sigma$. At the beginning of the game, each player then observes the choice of both. Thereafter the player faces the following sequential decisions:
    \begin{enumerate}
        \item What report to submit to $\Sigma$
        \item Whether to ratify or veto $\pi$
        \item Which action to take in the continuation game that results from the collective decision in stage 2
    \end{enumerate} 
    Our assumption is that $\sigma$, the realization of the signaling function $\Sigma$, occurs after stage 2 together with the revelation of each player's participation choice.

    The above baseline setting nests the situation \emph{without} informational punishment---that is, $\Sigma$ as a constant function---as a special case. There, players can only communicate before playing the game through their decision in stage 2. 

    Informational punishment, in contrast, equips players with an additional communication device. That device uses the information provided by the players and commits to releasing a (garbled) version of it after stage 2. That is, while each player decides about her message before play, that message becomes public during the game play. Our setting relies on the following two ideas. Often, it is reasonable to assume that players can publicly verify that they do not intend to take part in a mechanism or, equivalently, that they did not communicate with the mechanism ex ante. Moreoever, it is also reasonable to assume that players  have access to devices that send messages with some delay while they cannot commit to ignoring useful information presented to them.\footnote{Various variants of the above basic setting are possible without any effect on the presented results, including simultaneous design of a grand mechanism $(\pi,\Sigma)$, individual signals $\Sigma_i$, simultaneous submission of veto decisions, and information, etc.}

    \paragraph{Vetoes and Informational Punishment.} Formally, informational punishment maps type reports to public messages. Importantly, informational punishment works outside the veto constraint. That is, we assume players \emph{can credibly verify} that they veto the mechanism through some commitment device but \emph{cannot credibly commit ex ante} to ignoring public information. The first assumption is the key premise for participation constraints to be a nontrivial friction in the design of the mechanism \citep[see][, chapter 6 for a comprehensive discussion of the participation issue absent this assumption]{myerson_game_1991}. The second assumption is weak in a world in which access to some information channels is almost costless. Economically, the assumption that vetoes are public means that parties can visibly choose not to participate in settlements or peace negotiations, not to become a member of a newly founded institution, or to refuse to take actions necessary to ratify a proposed mechanism. The assumption that players cannot ignore information, in turn, implies that if it is beneficial at the interim stage,\footnote{If it \emph{were} beneficial for a player to ignore information, it would be commonly known (in equilibrium) and the information would become nonstrategic. But then, that player would have an incentive to access the information.} players will listen to press conferences held by negotiators, statements made by the institutions, or information leaked on the internet.\footnote{In \citet{Balzer:15}, we provide further details on that potential extension. There we also show that business practices such as releasing beta versions of products can be used to provide informational punishment.} 
        
    In addition, we assume that players can send information garbles with some delay. In fact, although our model formally assumes the existence of the signaling device $\Sigma$ before the start of the game, we can equivalently allow for a decentralized device $\Sigma_i$ such that each player, privately informed about her type, creates a signaling device on her own. The reason is that in equilibrium, no player expects that signals are relevant with positive probability. Instead, signals act purely as a disciplining device. Moreover, in Section 3.2 we state sufficient conditions such that it is without loss that players choose their devices even after they learn who has vetoed the mechanism.
        
    We now turn to the assumptions we impose on the mechanism space $\Pi$ before discussing our interpretation of informational punishment.

    \paragraph{Mechanism Space.} Our approach separates informational punishment from the design of the mechanism in the classical sense. That is, informational punishment simply augments a given set of mechanisms. Indeed, our preferred interpretation is to think of the designer(s) of $\Sigma$ as distinct from the designer of $\pi$, although they do not have to be.

	Classical mechanism design would assume that the set of available mechanisms, $\Pi$, contains every resource-feasible decision rule $\pi$. In most real-world applications in which parties are offered an alternative to a status quo mechanism, however, institutional constraints limit the designer's potential offers. We want to emphasize that these constraints do not affect our results as long as Assumptions (i) and (ii) hold. 

	Combined, Assumptions (i) and (ii) ensure that the designer is not exogenously restricted to a smaller set of outcomes than the default game (Assumption (i)) nor forced to select a one-size-fits-all game for all type vectors (Assumption (ii)). A direct consequence of Assumption (i) is that the designer can, in principle, induce any communication equilibrium  of the default game by acting as a pure mediator between players \citep{Forges:93,myerson_game_1991}. A consequence of Assumption (ii) is that the designer can act as a switch that for some realizations of players' types $\theta$ imposes game $A$. For other realizations, however, the designer imposes game $B$. In particular, Assumption (ii) allows the designer to sometimes revert players to the default game (or a clone thereof) \emph{inside the mechanism}.\footnote{See \citet{Balzer:15a} for an example in which the \emph{optimal} mechanism takes exactly this structure. In that setting, the mechanism space is unbounded but by assumption the participation concerns we are dealing with in this paper are absent.}

	The two assumptions together empower the mechanism designer to propose real alternatives to the status quo for \emph{some} occasions without the need to waive the default game's structure for \emph{any} potential continuation game.

	\paragraph{Updating and Information Structure.} Finally, we comment on our notational choice to represent information structures and relate them to the perhaps more familiar direct updating notation. Our formulation of information structures as a set of (conditional) distributions over type vectors captures the higher-order common knowledge that players hold at each stage of the game. By the properties of PBE, players agree at some hierarchy level about the distribution of types. However, the conditional distribution players hold may differ depending on their private information and so may some of the higher-order beliefs they hold. A simple special case to familiarize oneself with our notation is to assume that each player's type is distributed identically and independently over a discrete set $\Theta_i$ and the probability of being type $\theta_i$ is $p(\theta_i)$. Then $I_i^0(\theta_i)=p(\theta_i)$ for all $i$. Moreover, for type profile $\theta_{-i}=\theta \setminus \theta_i$, \[I_{-i}^0(\theta_{-i}|\theta_i)= \frac{I^0(\theta_{-i},\theta_i)}{I_i^0(\theta_i)}= \frac{p(\theta_i)\prod\limits_{\theta_{j}\in \theta_{-i}} p(\theta_{j})}{p(\theta_i)}=\prod_{\theta_{j}\in \theta_{-i}} p(\theta_{j}).\] 
		In this setting, if player $i$ decides to veto the mechanism off the equilibrium path, the no-signaling-what-you-don't-know condition of PBE requires that player $i$'s belief at the beginning of the continuation game is identical to her prior belief. That is,
			\[I_{-i}^{V=i}(\theta_{-i}|\theta_i)=I^{V=\emptyset}_{-i}(\theta_{-i}|\theta_i)=I_{-i}^0(\theta_{-i}|\theta_i)=\prod_{\theta_j \in \theta_{-i} } p(\theta_{j}).\]

	All non-deviating players $ j \neq i$ form an arbitrary off-path belief about $i$, $q_i(\theta_i)$. Moreover, the remaining information structure satisfies for any $\theta_{-j}=\theta \setminus \theta_j$
			\[I_{-j}^{V=i}(\theta_{-j}|\theta_j)=I^{V=\{i,j\}}_{-j}(\theta_{-j}|\theta_j)=I_{-j}^0(\theta_{-j}|\theta_j)=  q_i(\theta_i)\prod_{\theta_k \in \theta_{-j}\setminus \theta_i}  p(\theta_{k}).\]

    With independent types, $q_i(\theta_i)$ is completely arbitrary. 
    More generally, when also allowing for correlations between types, PBE requires that the common-knowledge information structure $I^V$ after $V$ deviated has to ``make sense'' under the prior $I^0$; or using the notation of \citet{Bergemann:13}, $I^0$ must be an \emph{expansion} of $I^V$, which means $I^V \in \mathcal I^0$.

	After that update, players receive the signal realization, $\sigma$, knowing $\Sigma$. Again, we have to follow an updating procedure on our way from $I^V$ to $I^{V,\sigma}$. Formally, Bayes plausibility implies that $I^V$ is an expansion of any $I^{V,\sigma}$ for all realizations $\sigma$. In addition, for realizations on the equilibrium path, the probabilities with which the various information structures occur are such that the distribution of $I^{V,\sigma}$ constitutes a mean-preserving spread of $I^V$. That is, under signal realization $\sigma$, and a single veto by player $i$, any player $j$ updates as follows on a type profile $\theta_{-j} \setminus \theta_j$,  
		\[ \begin{split}
		    I_{-j}^{V=i,\sigma}(\theta_{-j}|\theta_j) &= \frac{Pr(\sigma|\theta_j,\theta_{-j},\Sigma )}{Pr(\sigma|\theta_j, \Sigma)}\frac{I^0(\theta_j,\theta_{-j}|\theta_i)}{I_j^0(\theta_j|\theta_i)}q_i(\theta_i)\\ &=\frac{Pr(\sigma|\theta_j,\theta_{-j},\Sigma )}{Pr(\sigma|\theta_j, \Sigma)}q_i(\theta_i){\prod\limits_{\theta_k \in \theta_{-j} \setminus \theta_i} p(\theta_k) } \\&= \frac{Pr(\sigma|\theta_j,\theta_{-j},\Sigma) }{Pr(\sigma|\theta_j,\Sigma) }I_{-j}^{V=i}(\theta_{-j}|\theta_j),
\end{split} \]

	where $\Sigma$ affects only $Pr(\sigma|\cdot)$. Importantly, the belief on a deviator, $q_i(\theta_i)$, is unaffected by $Pr(\sigma|\cdot)$ because players cannot rule out double deviations and instead form \emph{one} arbitrary off-path belief. Finally, averaging information structures $I_{-j}^{V=i,\sigma}$ using $\Sigma$ implies $I_{-j}^{V=i}$ since (on-path) beliefs are martingales.

	While the description of the potential updating paths is straightforward in our special case, correlation among types makes the notation in explicit terms cumbersome. Fortunately, when using the concepts of \citet{Bergemann:13}, all that matters for our results \cref{prop:full_participation,prop:refinement,prop:Inscrutability} is that the information structures $I^{V,\sigma}, I^V$, and $I^0$ are related through the expansion notion and that the Bayes plausibility constraint holds on the equilibrium path---a  well-defined property. For our result on informational opportunism, \cref{prop:IPBP}, we require full support of the correlated type space. Without it, the concept of expansions is too weak when allowing for informational opportunism.\footnote{See the section on correlated types in \citet{Fudenberg:88}.}

\section{Analysis} 
\label{sec:analysis}

\subsection{Main Result} 
\label{sub:main_result}

	An optimal full-participation mechanism might not exist absent informational punishment, even if $\Pi$ is arbitrarily large. To see this, take a situation in which no player vetoes the mechanism and check for potential deviations. A deviator $j$ who vetoes in that situation guarantees herself the prior belief about the other players' types via a deviation. At the same time, the deviator is most punished by the worst off-path belief assigned to her. If $j$'s outside option exhibits concavities in the information structure, it may be beneficial to make another player, $i$, veto the mechanism to relax $j$'s participation constraint.

	To illustrate why all unanimously ratified mechanisms might be suboptimal, assume that informational punishment is unavailable. That is, the signal $\Sigma$ is constant and thus provides no information. Further, suppose that only \emph{some types} of player $i$ are expected to veto the mechanism. These types then play the status quo game against the prior distribution of the remaining players, whereas players use equilibrium reasoning to update on player $i$'s type, resulting in an information structure $I^{V=i}$ conditional on a veto and an information structure $I^{V=\emptyset}$ conditional on no veto. Veto decisions are public. Therefore, player $j$ who vetoes the mechanism off the equilibrium path will learn in addition whether $i$ vetoed. PBE \emph{requires} $j$ to update about $i$. The outside option of $j$ changes with the on-path veto of $i$ compared to any full-participation mechanism. That change, in turn, affects the set of implementable outcomes. If $i$'s on-path veto decreases $j$'s outside option, the set of outcomes the designer can implement with on-path vetoes is larger than that without on-path vetoes.

	With informational punishment---that is, a nonconstant $\Sigma$---the designer can relax $j$'s participation constraint without relying on on-path rejections and (possibly) beyond what is implementable with rejections. Veto equilibria always require that vetoes be incentive compatible. In our example above, $i$'s incentive compatibility further restricts the set of attainable mean-preserving spreads over $I^0$. Informational punishment can threaten deviators by implementing \emph{any} mean-preserving spread.

	\begin{proposition}\label{prop:full_participation}
		It is without loss of generality to focus on mechanisms that ensure full participation if informational punishment is available.
	\end{proposition}

 \begin{proof}
    The proof is constructive. Take any $\pi$, and a \emph{veto equilibrium} in which the mechanism is vetoed with positive probability on the equilibrium path. We first characterize the decision rule of the veto equilibrium. Then, we show it can be implemented with full participation using \emph{informational punishment}.
		 	
    Let $\xi(\theta)$ be the probability that $\pi$ is vetoed given type profile $\theta$. Moreover, $\xi_i(\theta_i)$ is the likelihood that type $\theta_i$ vetoes $\pi$ on the equilibrium path. If players mix regarding their veto decision, the set of players that vetoed, $V$, might be random. Let $\mathscr{V}$ be the set of all sets $V$ which occur with positive probability. After an on-path veto, players then observe $V \in \mathscr{V}$ and update to information structure $I^{V}$. Outcomes realize according to $\mathring{\pi}_{I^{V}}$. Taking expectations over all realizations of $V$, the ex-ante expected continuation game conditional on a veto is a lottery $(P(V), \mathring{\pi}_{I^{V}})$ defined over all $V \in \mathscr{V}$. $P(V)$ is the on-path likelihood that a veto is caused by the set $V$ and not by any other set. Because $\mathring{\Pi} \in \Pi$ and $\Pi$ is closed under convex combinations, the lottery implies $\sum P(V) \mathring{\pi}_{I^{V}} \in \Pi$.

    Conditional on no veto, the information structure is $I^a$, and $\pi_{I^a}$ is the decision rule.
			
    The grand game implements an $I^0$-IC decision rule $\pi'_{I^0}:= \xi \sum P(V) \mathring{\pi}_{I^{V}} + (1-\xi) \pi_{I^a}$. Again, $\pi'\in \Pi$ because $\Pi$ is closed under convex combinations.

    We now construct a signal $\Sigma$ such that the mechanism $\pi'_{I^0}$ is implementable under full participation. By construction, $\pi'$ is feasible and $I^0$-IC. What remains is to show that no player has an incentive to veto $\pi'$.

    We construct the following signaling device $\Sigma_i: \Theta_i \rightarrow \Delta(\{0,1\})$ where $\sigma_i(\theta_i)=1$ with probability $\xi_i(\theta_i)$ and $0$ otherwise. 
    When observing off-path behavior (i.e., a veto) by $i$, $j$ believes that $i$ has randomized uniformly over the entire type-space when reporting to $\Sigma_i$. Thus, she disregards $\sigma_i$. We choose the off-path belief on $i$ identical to the belief attached to $i$ after observing her unilateral veto in the veto equilibrium.

    No player $i$ has an incentive to veto the mechanism. If a player vetoes the mechanism, $\Sigma$ provides her with the same lottery over information structures that she expects from a veto in the veto equilibrium. Participation, in turn, gives the same outcome as the veto equilibrium. No player can improve the outcome of the veto equilibrium by vetoing $\pi'$. 

    Truthful reporting to $\Sigma_i$ is a best response. $\Sigma_i$ is on-path payoff-irrelevant. Thus, under $(\pi', \Sigma)$ an equilibrium with full participation in $\pi'$ exists that implements the same outcome as the veto equilibrium.
\end{proof}

\subsubsection*{Optimal Informational Punishment}

	Suppose the task is to design the optimal mechanism under some objective. \Cref{prop:full_participation} shows that a mechanism with informational punishment and full participation can replicate the outcome of any equilibrium with on-path rejection. However, it may not be immediately clear from \cref{prop:full_participation} how that property simplifies the mechanism-design problem. Here, we outline the implied simplification. Informational punishment is most effective if the signaling device conditions on the deviator's identity. That is, $\Sigma^V: \times_{j \in N \setminus V} \theta_j   \rightarrow \Delta(S)$ is the signaling device used if the set of vetoing players is $V$. Specifically, $\Sigma^i$ is the mapping from reports to realizations used if (only) $i$ vetoes the mechanism.

	Informational punishment separates the participation problem from the mechanism-design problem. The reason is straightforward in light of the proof of \cref{prop:full_participation}. Informational punishment affects the participation constraints only. Moreover, by \cref{prop:full_participation} it is without loss to restrict attention to optimal mechanisms in which vetoes are off-path events. Therefore, we can limit our attention to unilateral vetoes when constructing the optimal mechanism. 

	What remains is to find the $\Sigma^i$ that punishes deviator $i$ the most. If player $i$ vetoes, all other players hold an off-path belief about player $i$'s type and use information structure $I^i$ to update any signal sent by $\Sigma^i$. Let $P(I)\in [0,1]$ be the probability that information structure $I\in \mathcal I^0$ arises after player $i$ vetoes the mechanism. To solve for the optimal $\Sigma^i$, we can solve for the lottery over information structures $(P(I), I)$ such that $\sum_I P(I)=1$. Formally, we solve the following information-design problem:\footnote{See, for example, \citet{bergemann:16} and references therein for more information on information-design problems.}

	\begin{eqnarray*}		
		\min_{P(I)} \sum_I P(I) \sum_{\theta_i} \alpha(\theta_i) v_i(\theta_i,I,\mathring{\pi}_I) \\
		s.t. \sum_I P(I) I= I^i, \text{and} \\
		I_{i}(\theta_i|\theta_{-i} )=I^{V=i}_{i}(\theta_i|\theta_{-i}).
	\end{eqnarray*}
	Here, $\alpha(\cdot) \in [0,1]$ with $\sum_{\theta_i} \alpha(\theta_i)=1$ are weights corresponding to the Lagrangian multipliers of the binding participation constraints in the mechanism-design problem (see \citet{Jullien00}).\footnote{In many standard environments, it holds that $\alpha(\theta_i)=1$ for some $\theta_i$; that is only the participation constraint of one type binds. Moreover, for many default games, the solution to the information-design problem is independent of $\alpha$. Such a situation occurs in settings in which there is a worst information structure for all types of the deviator. This property holds, in particular, in games that feature strategic complements or strategic substitutes.}
	The first constraint implies that every $I$ results from some feasible signal $\Sigma^i$: the signal is \emph{Bayes plausible.} The second constraint means that $\Sigma^i$ cannot reveal information about the deviator who vetoed the mechanism: the belief about the deviator, $I^{V=i}_{i}(\theta_{i}|\theta_{-i})$, is constant. Players do not observe whether the deviator has deviated only at the ratification stage or before that. Therefore, they attach a single belief to the deviator's type distribution.

	The solution to the information-design problem relaxes player $i$'s participation constraint the most. Thus, we can determine each player's least binding participation constraint, player by player. Once participation constraints under informational punishment are determined, we can design the mechanism taking each player's participation constraint as exogenously given. Using the same arguments as those in the concavification literature \citep{Aumann:95}, optimal informational punishment \emph{convexifies} the deviation payoffs and thus minimizes the gains from a veto.

\subsection{Other Environments} 
\label{sub:other_environments}
	
	This section shows that our results extend straightforwardly to more complex settings. We begin by looking at refined equilibrium concepts. Then we consider informed-principal problems. Finally, we reduce the designer's commitment power. 

\subsubsection*{Refinements}

	\Cref{prop:full_participation} assumes that the designer can freely pick off-path beliefs under the PBE restriction. She can select any first-node off-path belief of the continuation game. Depending on the context and the application, such an equilibrium selection might not be reasonable. Using a refinement could instead make on-path vetoes unavoidable because it limits the designer's equilibrium choice set in the first place.\footnote{\Citet{correia2017trembling} provides additional discussion of this issue and how it may interfere with the design of a mechanism.}

	Our second finding is that \Cref{prop:full_participation} is robust to most common refinements. Specifically, whenever we refine the equilibrium concept according to 

\begin{equation*}
	(\star) \in \lbrace \text{Perfect Sequential Equilibrium, Intuitive Criterion, Ratifiability} \rbrace,
\end{equation*}

	full participation remains optimal.

\begin{proposition}\label{prop:refinement}
	Suppose the solution concept is perfect Bayesian equilibrium with refinement concept $(\star)$, and informational punishment is available. Then, focusing on mechanisms that imply full participation is without loss of generality.
\end{proposition}
\begin{proof}
				Ratifiability requires full participation in the mechanism and therefore holds trivially, as the designer can always choose a degenerate signaling device. It thus is without loss of generality to show full participation under refinement  $(\star)' \in$ \{Perfect Sequential Equilibrium, Intuitive Criterion\}.
		
				Consider the veto equilibrium used in the proof of \Cref{prop:full_participation}. Suppose this equilibrium satisfies refinement $(\star)'$. 

				We show that the full-participation equilibrium constructed in the proof of \Cref{prop:full_participation}, $(\pi',\Sigma)$, satisfies the same refinement criterion. Two aspects are crucial. First, compare the equilibrium with vetoing and that with full participation. On-path (expected) outcomes and those that are off-path but can be reached by a unilateral deviation are identical between these two equilibria for every state $\theta$. First, take any state $\theta$ in which the mechanism is unanimously accepted in both equilibria. Then both outcomes coincide, and so does the credibility of the beliefs. Second, consider a state $\theta$ in which the mechanism is rejected in the veto equilibrium with positive probability. For the same state, suppose that $\pi'$ is rejected in the full-participation equilibrium---an off-path event. The resulting  \emph{off-path belief} on the deviator $\theta_i$ coincides with the \emph{on-path belief} on the same $\theta_i$ in the veto equilibrium.

				Thus, the constructed off-path beliefs put positive mass only on those types that \emph{weakly prefer to deviate}, while no such type \emph{strictly prefers to deviate}. Thus, any off-path belief for type $\theta_i$ is credible in the sense of \citet{grossman1986perfect}, \emph{and} off-path beliefs do not violate the intuitive criterion.
			\end{proof}

\subsubsection*{Informed-Principal Problems}
	
	The informed-principal environment assumes that a privately informed player proposes the mechanism that should replace the status quo.

	Formally, instead of a nonstrategic third party, one of the players, say $i=0$, proposes a mechanism as an alternative to the default game. The setting becomes an informed-principal problem. Players $i=1,...,N$ are the agents. 

	A key concept to solve informed-principal problems is the concept of inscrutability (see \cite{Myerson:83informed}). It states that it is without loss to assume that the informed principal, player $0$, selects a mechanism that does not allow the other players $1,...., N$ to learn about the principal's type from the proposed mechanism. That is, inscrutability means it is without loss to restrict attention to pooling solutions in which each principal type offers the same mechanism.

	The default game can depend nonlinearly on beliefs about player $0$'s type. Consequently, the principle of inscrutability might fail. Player $0$ may have strict incentives to signal private information via the mechanism proposal. Player $0$ thereby relaxes the other players' participation constraints. The following result states that these concerns are irrelevant if informational punishment is available.

\begin{proposition} \label{prop:Inscrutability}
	The principle of inscrutability holds if informational punishment is available.
\end{proposition}
			\begin{proof}
				Consider an equilibrium of the grand game such that different types of player $0$ propose different mechanisms. Let $\mathcal{M}$ be the set of $\pi$s that are proposed with strictly positive probability. Let $\xi^{ \pi}_0(\theta_0)$ denote the probability that player 0 type $\theta_0$ proposes mechanism  $ \pi  \in \mathcal{M}$.
			
				Consider the case in which at least one type of one player vetoes some $\hat{\pi}   \in \mathcal{M}$ on the equilibrium path. We refer to this equilibrium as the separate-and-veto equilibrium. Recall that if $\hat{\pi}$ is vetoed, some rule in $\mathring{\Pi}$ results. Let the probability that $ \hat{\pi}$ is vetoed be $\xi^{\hat{\pi}}$. Moreover, $\xi^{\hat{\pi}}_i(\theta_i)$ is the probability that $\theta_i$ vetoes $\hat{\pi}$. The separate-and-veto equilibrium implements an $I^0$-IC decision rule, $\overline{\pi}_{I^0}$. Because $\Pi$ is closed under convex combinations $\overline{\pi}_{I^0} \in \Pi$.

		  		We prove the existence of the following equilibrium. All types of player $0$ propose $\overline{\pi}_{I^0} $ and every player accepts it. This equilibrium leads to the $I^0$-IC decision rule $\overline{\pi}_{I^0}$. We construct a signaling device $\Sigma$ to support acceptance of $\overline{\pi}_{I^0} $.  Let $o: \mathcal{M} \rightarrow \mathbb{R}$ be some invertible function. For $i=0$, we construct the signal $\Sigma_{0}(\theta_0)$ with support $ \lbrace  o(\pi') \rbrace_{ \pi' \in \mathcal{M}}$ and associated probabilities $Pr(\pi'|\theta_0)=\xi^{\pi'}_0(\theta_0)$. For any $i >0$, let the signal be $\Sigma_{i}(\theta_i) = 1$ with probability $ \xi_i(\theta_i)$ and $\Sigma_{i}(\theta_i) = 0$ with remaining probability. Whenever player $i>0$ vetoes, a signal realizes according to $\Sigma $. Thus, the reason why no player rejects $\overline{\pi}$ is the same as in the proof of \Cref{prop:full_participation}. The only difference is that the signaling function also replicates the potential signal-by-mechanism-choice behavior of the principal, captured by $\Sigma_0$.
			\end{proof}
\subsubsection*{Informational Opportunism}

	Our design approach on the side of the signal assumes that the signaling device $\Sigma$ is fixed. In other words, the signal designer has commitment power. However, if the signal is created through passing information to a third player---for example, a journalist---then that assumption does not always hold. In such a setting, the person that keeps the information may be a strategic player. In addition, she could be unable to both commit to the signaling device and report its outcome. \Citet{MartimortDequiedt15} refer to such a setting as \emph{informational opportunism}.

	To adequately address the situation with informational opportunism, we need to take a stance on two aspects. First, how large is the signal designer's commitment power? Second, does the designer's objective change once she sees a deviation? Is it credible for her to use the information to punish the deviator? 

	An extreme form of informational opportunism occurs if the signal designer chooses the \emph{signal realization}, $\sigma$, after receiving a report rather than the \emph{signaling device}, $\Sigma$. In that case, the signal cannot commit to a mapping from reports to realizations. Instead, the signal becomes a cheap-talk announcement. The results of \Cref{prop:refinement,prop:full_participation,prop:Inscrutability} trivially do not hold without any commitment power.

	Instead, we focus on a signal designer whose action space is still the set of signaling devices. We make the following assumption. 

\begin{assumption}[No Fabricated Data]\label{ass:nofab}
	The choice of signaling device, $\Sigma$, becomes public together with its realization, $\sigma$.
\end{assumption}

	Under \Cref{ass:nofab} the signal designer can back up her claims by providing evidence. Indeed, all interested players can see the designer's method to reach her conclusion, $\sigma$. Within this setting, informational opportunism is best seen when considering the timing of the grand game. In our baseline model, the signal designer commits to $\Sigma$ \emph{before} players decide about vetoing. She is thus an impartial third player and not an interested player in the grand game. In contrast, we assume now that the designer is an interested player in the game. 

	First, we assume that the signal designer can commit to an objective. However, she \emph{cannot} commit to her signaling device at the beginning of the game. Instead, she picks $\Sigma$ after players have made their acceptance decision and the designer has elicited the information. That is, the signal designer faces \emph{ex post incentive constraints.}

\begin{definition}[Informational Opportunism]
	The designer of the signaling device suffers from \emph{informational opportunism} if she cannot commit to a signaling device $\Sigma$ before the players' participation decision.
\end{definition}

	We are interested in situations in which our previous results are \emph{immune} to informational opportunism.

\begin{definition}[Immunity]
	A result is \emph{immune to informational opportunism} if it is implementable by a signal designer that suffers from \emph{informational opportunism}.
\end{definition} 

	Allowing for informational opportunism comes at a cost. We cannot make a general statement on immunity. However, we state a set of definitions that restricts the environment. These restrictions allow us to state \Cref{prop:IPBP}, which determines a condition for immunity.

	We identify a signal designer by her \emph{type}. The designer's type is the information she has elicited from the (participating) players.
	Observe that any signal designer \emph{could} fully reveal her type by choosing the appropriate signaling device.

	Yet if the signal designer chooses not to reveal her type, the interpretation of realization $\sigma$ does not depend only on $\Sigma$. It also depends on the belief that players form about the signal designer. To form the belief, players observe $\Sigma$. Each $\Sigma$ triggers a belief that, together with $\Sigma$, leads to a lottery over information structures. Different designer types potentially have different preference rankings over lotteries for a given signal designer's objective. We assume, however, that preferences are aligned, and all types share the same order. Thus, there is a common understanding of which information structures better achieve the desired goal.

	To achieve immunity to informational punishment, we impose two properties on the environment. First, we assume \emph{full support} of the type distribution, $I^0>0$.
	Second, the signal designer has \emph{aligned preferences} across her types.
\begin{definition}[Aligned Preferences]
	Fix arbitrary distributions over a collection of $(N-1)$ players' types. Let $F$ and $F'$ be two (possible) distributions of the remaining player $i$'s type. The signal designer has aligned preferences if every designer type prefers $F$ to $F'$ whenever $F$ first-order stochastically dominates $ F'$.
\end{definition}

	Finally, we define an extreme notion of the desire to separate.

\begin{definition}[Unraveling Pressure]\label{def:unrav}
 	A signal designer faces unraveling pressure under signaling device $\Sigma$ if she strictly prefers to verify her type than to engage in the lottery induced by $\Sigma$.
\end{definition}

	Suppose full support and aligned preferences. In that case, the absence of unraveling pressure is necessary and sufficient to guarantee that a signaling device is implementable, as the following proposition shows.
\begin{proposition}\label{prop:IPBP}
 	Suppose $I^0>0$ and the signal designer's preferences are aligned. Then, a signaling device $\Sigma$ is implementable under informational opportunism if and only if no signal-designer type faces unraveling pressure.
\end{proposition}
	\begin{proof}
		The ``only if'' direction follows from \cref{def:unrav}. If a signal designer type faces unraveling pressure, she prefers to reveal her type over the signaling device $\Sigma$.

		For the ``if'' direction, consider a signaling device $\Sigma$. Assume player $i$ has vetoed the mechanism. The signal designer elicited the information $\theta_{-i}$ from the non-deviating players. We want to show that no designer type, $\theta_{-i}$, has an incentive to announce a different device than $\Sigma$.

		Suppose signal designer type $\theta_{-i}$ deviates by announcing $\Sigma'$ which does not verify $\theta_{-i}$. Players observe the deviation $\Sigma'$ and its realization $\sigma'$. Using these objects, they form off-path beliefs about the types of all $N-1$ players. The symmetry of PBE and the full-support assumption imply the following. Any subset of players has identical beliefs about those not in that subset.

		The off-path beliefs on the signal designer's type are only restricted by the signaling function $\Sigma'$. If a realization $\sigma'$ occurs with probability $0$ given a type $\theta_{-i}$, then players exclude that type from the set of possible signal designer types. Denote the set of not excluded types by $\Theta^{\sigma'}$. The distribution $F^{\sigma'}: \Theta^{\sigma'} \rightarrow [0,1]$ is arbitrary. That is, for every $\Sigma'$, there always exists an off-path belief about the deviating designer that rationalizes $F^{\sigma'}$.

		By assumption $\Sigma'$ does not verify $\theta_{-i}$. Thus, $|\Theta^{\sigma'}|>1$. Types have aligned preferences. Thus, we can find a signal designer type $\tilde{\theta}$ such that a degenerate belief on $\tilde{\theta}$ makes every designer type other than $\tilde{\theta}$ worse off compared to that signal designer revealing her type. No unraveling pressure implies that no type benefits from the deviation. $\Sigma$ is implementable under informational opportunism.
	\end{proof} 
    The immunity of our results to informational opportunism is a straightforward corollary to \cref{prop:IPBP}. Moreover, it provides a simple way to test for immunity given a candidate $\Sigma$.

\begin{corollary}
	Suppose players' types are independently distributed and the designer's preferences are aligned.  \Cref{prop:full_participation,prop:refinement,prop:Inscrutability} are immune to informational opportunism if no signal-designer type faces unraveling pressure under $\Sigma$.
\end{corollary}

	An essential requirement of Proposition~\ref{prop:IPBP} is that the signal designer's preferences be aligned. If preferences are misaligned, informational opportunism can eliminate the power of informational punishment completely. We show this using the following example.

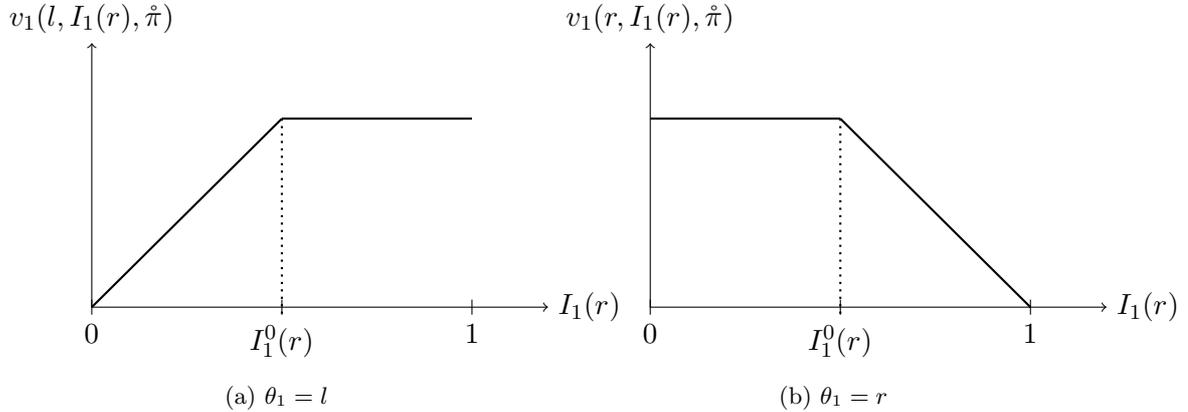
\begin{figure}[thb]
    \centering
    \begin{subfigure}[t]{.5\textwidth}
      	
\begin{tikzpicture}[scale=5]
\node [above] at (0,.7) {$v_1 (l, I_1(r), \mathring{\pi}  )$};
\node [right] at (1.2,-.0) {$  I_1(r)$};

\draw [<->] (0,.7) -- (0,0) -- (1.2,0);

\draw (0,0.02)--(0,-0.02) node[below] {$0$};
\draw (1,0.02)--(1,-0.02) node[below] {$1$};
\draw (0.5,0.02)--(0.5,-0.02);
\draw[dotted, thick] (0.5,0.5) -- (0.5,-0.02) node[below] {$I_1^0(r)$};

\draw[thick] (0,0)--(0.5,0.5);   
\draw[thick] (0.5,0.5)--(1,0.5);

\end{tikzpicture} %
     \caption{$\theta_1=l$}
    \end{subfigure}%
    \begin{subfigure}[t]{.5\textwidth}
          
\begin{tikzpicture}[scale=5]
\node [above] at (0,.7) {$v_1(r, I_1(r),\mathring{\pi}  )$};
\node [right] at (1.2,-.0) {$  I_1(r) $};

\draw [<->] (0,.7) -- (0,0) -- (1.2,0);

\draw (0,0.02)--(0,-0.02) node[below] {$0$};
\draw (1,0.02)--(1,-0.02) node[below] {$1$};
\draw (0.5,0.02)--(0.5,-0.02);
\draw[dotted, thick] (0.5,0.5) -- (0.5,-0.02) node[below] {$I_1^0(r)$};

\draw[thick] (0,0.5)--(0.5,0.5);   
\draw[thick] (0.5,0.5)--(1,0);

\end{tikzpicture} 
 %
        \caption{$\theta_1=r$}
    \end{subfigure}%
    \caption{The left (right) panel depicts $\theta_1=l$'s ($r$'s) payoff from the noncooperative interaction with player 2 in the default game under belief $\beta$.}
    \label{fig:Player  1}
\end{figure}

\begin{example}
	There are two players, 1 and 2, and only player 1 has private information $\theta_1 \in    \lbrace l, r\rbrace$. We denote player 2's belief about $\theta_1=r$ by $I_1(r) $ with commonly known prior $I_1^0(r)$. Let the expected payoff from playing the status quo under belief $I_1(r)$ be $v_1(\theta_1,I_1(r),\mathring{\pi})$ for player 1 and $v_2(I_1(r),\mathring{\pi} )$ for player $2$. 
	\cref{fig:Player  1} depicts player 1's payoff resulting from their interaction in the default game for an arbitrary belief $I_1(r)$. Both types of player 1 benefit if player 2 mistakenly believes her to be the opposite type. Preferences are thus misaligned.\footnote{The situation is a reduced-form description of a game in which player 1 is an incumbent firm that privately knows which technology, $\theta_1$, is demanded. Player 2 joins the market and decides how to design a competing product. This choice involves  $a_2 \in [0,1]$, which determines player 2's technology $\theta_2 := a_2 r + (1-a_2)l$. The closer player 2's technology is to that of player 1, the fiercer the competition and the lower the rents. A monopolist receives rents of 0.5. }

	For simplicity, we assume that the designer's optimal mechanism is such that player 2's outside option is binding. Moreover, we consider a decentralized implementation of informational punishment. That is, player 1 sets up the signaling device $\Sigma$ herself. \Cref{fig:Player  2} plots player 2's payoffs and shows that an optimal signaling device, $\Sigma^*$, decreases player 2's outside option down to zero. Indeed, if $Pr(\sigma_r|r,\Sigma^*)=1$ and $Pr(\sigma_l|l,\Sigma^*)=1$, it holds that player 2's expected payoff is $I_1^0(r) v_2(1,\mathring{\pi} )+(1-I_1^0(r))v_2(0, \mathring{\pi} )=0$. No type of player 1 is subject to unraveling pressure, as $\Sigma^*$ is fully revealing.

	If player 1 commits to $\Sigma^*$ before observing a veto by player 2, player 2 joins the mechanism as long as she receives a non-negative expected payoff. Now suppose that player 1, in contrast, submits $\Sigma^*$ \emph{after} having learned that player 2 deviates. In that situation, it is not an equilibrium of the continuation game that player 1 submits $\Sigma^*$. 

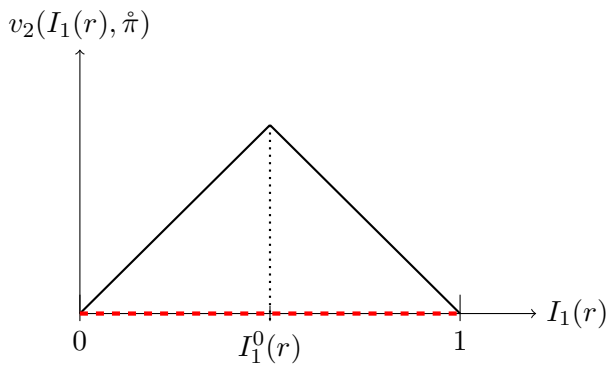
\begin{figure}[bht]
    \centering
      
\begin{tikzpicture}[scale=5]
\node [above] at (0,.7) {$v_2(I_1(r),\mathring{\pi}  )$};
\node [right] at (1.2,-.0) {$ I_1(r) $};

\draw [<->] (0,.7) -- (0,0) -- (1.2,0);

\draw (0,0.05)--(0,-0.02) node[below] {$0$};
\draw (1,0.05)--(1,-0.02) node[below] {$1$};
\draw (0.5,0.02)--(0.5,-0.02);
\draw[dotted, thick] (0.5,0.5) -- (0.5,-0.02) node[below] {$I_1^0(r)$};

\draw[thick] (0,0)--(0.5,0.5);   
\draw[thick] (0.5,0.5)--(1,0);
\draw[ultra thick, red, dashed] (0,0)--(1,0);

\end{tikzpicture} 
    \caption{The solid line depicts player 2's payoff from the default game under belief $I_1(r)$. The dashed line depicts player 2's expected payoff under the payoff-minimizing $\Sigma^*$, which fully reveals player 1's type.  }
    \label{fig:Player  2}
\end{figure}

	To see this, suppose---for a contradiction---it was  and consider the following deviation to a fully uninformative signal: player 1 chooses a ``babbling'' signaling device that sends the same signal $\sigma'$ independent of her type. To make that deviation unprofitable for $\theta_1=l$, we need to ensure that player 2's off-path belief after that deviation, $I_1^{\emptyset}(r)$, equals $0$. However, under that off-path belief, the initial deviation is profitable for type $\theta_1=r$. This is a contradiction.

	Thus, the threat that player 2 will learn player 1's type when vetoing the mechanism  (that is, $\Sigma^*$ is active) is not credible if player 1 cannot commit to $\Sigma^*$ before learning about player 2's veto. As a consequence, player 2 requires a larger outside option than the expected payoff under $\Sigma^*$ to participate in the mechanism. In fact, here, the only signaling device that is immune to informational optimism is the babbling signaling device, leaving player 2 with prior beliefs after a veto.\hfill$\blacktriangle$ 
\end{example}

	The other requirement, full support, allows us to define PBE after the first type of deviation (veto) such that we can account properly for the second type of deviation (designing the signal) within the continuation game.

	We conclude our discussion by addressing a weaker notion of informational opportunism.

	The signal designer may commit to a signaling device $\Sigma$ but may choose to conceal the realization. We refer to this as weak informational opportunism.

\begin{definition}[Weak Informational Opportunism]
	The designer of the signaling function $\Sigma$ suffers from \emph{weak informational opportunism} if she can commit to a signaling function $\Sigma$ at the beginning of the game but not to the disclosure of the realization $\sigma$.
\end{definition}

	It is straightforward to see that immunity to informational opportunism implies immunity to weak informational opportunism. In addition, we can drop the no-unraveling-pressure condition. The reason is that---with aligned preferences---there is a common worst realization $\underline{\sigma}$ across signal-designer types. If a signal designer hides information off the equilibrium path, an off-path belief assuming a (hidden) realization of $\underline{\sigma}$ punishes every designer type the most. Consequently, using the standard unraveling arguments from the persuasion literature \citep[see, for example,][]{Milgrom:81,grossman1981}, no designer has an incentive to hide her information. Moreover, the result applies even if players are unaware of the signal designer's objective.

\begin{corollary}\label{cor:ambiguity}
	Suppose $I^0>0$ and the designer's preferences are aligned. \Cref{prop:full_participation,prop:refinement,prop:Inscrutability} are immune to weak informational opportunism even when players face ambiguity over the signal designer's objective function.
\end{corollary} 
	The reason for \cref{cor:ambiguity} is that if the signal designer's preferences are aligned, then given any objective, $\Sigma$ has a common worst signal realization. Not revealing the signal realization leads to an arbitrary off-path belief. Thus, even if the designer's objective is unclear, parties can coordinate on an off-path belief in some PBE that puts all probability mass on the worst signal.

\section{Final Remarks} 
\label{sec:conclusion}

	Mechanism design can facilitate good policy making. It provides a simple benchmark that informs our understanding of what is possible theoretically. The power of mechanism design derives from its simplicity in calculations. Invoking the revelation principle, we can derive strong results even in a complex environment. 

	However, suppose the environment is such that the designer cannot control the entire strategic interaction, and parties can block the implementation of the mechanism. In that case, the revelation principle for the part the designer controls can fail. The reason is that parties can use their veto power to signal private information strategically. To obtain the desired benchmarks, we would have to either restrict the environment to settings in which signals through vetoes are nonprofitable or delve into complex case distinctions. 

	We argued that restricting to the setting in which strategic vetoes are of no concern is without loss provided that parties have access to a tool we call \emph{informational punishment}. Informational punishment allows parties to store information for some time and release a garbled version of it in case of a deviation. Furthermore, we show that through informational punishment, we can transform every environment in which strategic vetoes are relevant into an equivalent setting in which they are not. Thus, we can---without loss---restrict ourselves to full-participation mechanisms under an (appropriate) outside option.

	Our results go beyond classical applications of mechanism design. We derive a minimum condition for the available mechanism space that suffices to guarantee full participation at an optimum when informational punishment is available. Informational punishment works off the equilibrium path, does not affect incentive compatibility directly, and allows for publicly verifiable rejections. Informational punishment can be implemented through a centralized signal, through the designer of the mechanism (who could also be an informed principal), or decentralized through the parties individually. Furthermore, informational punishment is robust to various additional constraints on the setting.

\renewcommand{\partname}{}
\renewcommand{\thepart}{}
\begingroup
\setlength\bibitemsep{5pt}
\part{References}
\begin{refcontext}[sorting=nyt]
\printbibliography[heading=none]
\end{refcontext}
\endgroup
\end{document}